\documentclass[12pt,twoside,a4paper,leqno]{article}
\usepackage[pdftex,%
pdftitle={Energy operator and scalar curvature in geometric quantization},%
pdfauthor={C. Tejero Prieto, R. Vitolo},%
pdfkeywords={Classical mechanics, quantum mechanics, geometric quantization},%
]{hyperref}
\usepackage{amsmath}
\usepackage{amssymb}
\usepackage{amsthm}
\usepackage{xcolor}
\usepackage{amscd}

\newtheoremstyle{MyThm}
  {3pt}
  {3pt}
  {\itshape}
  {\parindent}
  {\bfseries}
  {.}
  {.5em}
  {}
\swapnumbers
\theoremstyle{MyThm}
\newtheorem{Definition}{Definition}[section]

\newtheorem{Corollary}[Definition]{Corollary}

\newtheorem{Lemma}[Definition]{Lemma}

\newtheorem{Proposition}[Definition]{Proposition}
\newtheorem{Remark}[Definition]{Remark}
\newtheorem{Theorem}[Definition]{Theorem}

\newcommand{\C}{\mathbb{C}}

\newcommand{\Z}{\mathbb{Z}}

\newcommand{\cin}{C^\infty}

\newcommand{\cC}{\mathcal{C}}

\newcommand{\cO}{\mathcal{O}}

\newcommand{\cH}{\mathcal{H}}

\newcommand{\olnabla}{\overline{\nabla}}

\newcommand*{\pd}[2]{\mathchoice{\frac{\partial#1}{\partial#2}}
  {\partial#1/\partial#2}{\partial#1/\partial#2}
  {\partial#1/\partial#2}}

\DeclareMathOperator{\End}{End}

\DeclareMathOperator{\Hom}{Hom}

\DeclareMathOperator{\Id}{Id}

\DeclareMathOperator{\dive}{div}

\title{\bf The geometry of\\ real reducible polarizations\\
   in quantum mechanics}
\author{
\bf Carlos Tejero Prieto$^1$, Raffaele Vitolo$^2$
\bigskip
\\
\footnotesize $^1$ Departamento de Matematicas, Universidad de Salamanca
\\
\footnotesize and IUFFyM
\\
\footnotesize Pl. de la Merced 1--4, 37008 Salamanca, Spain
\\
\footnotesize email: {\tt carlost@usal.es}
\medskip
\\
\footnotesize $^2$Dipartimento di Matematica e Fisica ``E. De Giorgi,
Universit\`a del Salento
\\
\footnotesize and INFN -- Sezione di Lecce
\\
\footnotesize Via per Arnesano, 73100 Lecce, Italy
\\
\footnotesize email: {\tt raffaele.vitolo@unisalento.it}
}
\date{\today}

\begin{document}
\maketitle
\begin{abstract}
  The formulation of Geometric Quantization contains several axioms and
  assumptions.  We show that for real polarizations we can generalize the
  standard geometric quantization procedure by introducing an arbitrary
  connection on the polarization bundle. The existence of reducible quantum
  structures leads to considering the class of Liouville symplectic
  manifolds. Our main application of this modified geometric quantization
  scheme is to Quantum Mechanics on Riemannian manifolds. With this method we
  obtain an energy operator without the scalar curvature term that appears in
  the standard formulation, thus agreeing with the usual expression found in
  the Physics literature.


\end{abstract}

\section{Introduction}

Quantization of Classical Mechanics is a procedure introduced by Dirac that
starting from a classical observable yields a quantum operator. It is well
known that it is not possible to consistently quantize all observables $f(q,p)$
on a classical phase space, as this would lead to physical and mathematical
contradictions (see \emph{e.g.}  \cite{Sni80,Woo92}). A \emph{polarization} is
a mathematical structure that allows a coordinate-free definition of a subset
of quantizable observables and a Hilbert space of wave functions on which the
quantum operators act.

Polarizations can be real or complex. The latter have attracted the interest of
more researchers mainly due to two reasons. In first place complex
polarizations exist even when real polarizations cannot exist (consider, for
instance, the symplectic manifold $S^2$) and in second place they have
applications to representation theory~\cite{GuSt82}. During our research we
have realized that a class of real polarizations, the reducible polarizations,
have a geometry that is richer than what appears in the standard expositions of
quantization. In this paper we will study the geometry of real reducible
polarizations, and discuss some interesting consequences for the quantization
process.

Our framework is that of Geometric Quantization (GQ, \cite{Sni80,Woo92}), where
we have a classical Hamiltonian system on a symplectic manifold $(M,\omega)$,
$\dim M = 2n$. A real polarization is an integrable Lagrangian distribution
$P\subset TM$.  The choice of a polarization is the infinitesimal analogue of
the choice of a set of (locally defined) Poisson--commuting observables.  There
are well-known existence theorems concerning polarizations and other geometric
structures in GQ; the theorems impose restrictions on the topology of the
underlying symplectic manifold.

\emph{Quantum states} are represented by sections $\psi\colon M\to \mathcal L$
of a Hermitian complex line bundle $\mathcal L\to M$ endowed with a Hermitian
connection $\nabla^{\mathcal L}$ whose curvature is proportional to the
symplectic form $\omega$.

\emph{Quantizable observables} $\mathcal{O}_P$ are then just those classical
observables $f\in\cin(M)$ whose Hamiltonian vector field $X_f$ preserves the
polarization. Such vector fields act on \emph{polarized} sections $\psi$, that
is sections which are (covariantly) constant along the polarization, thus
providing a candidate for quantum operators. However the scalar product $\int
\psi_1\overline{\psi_2}$ of two wave functions $\psi_1$, $\psi_2$ is
ill-defined since the space of leaves $M/P$, even if it exists as a Hausdorff
manifold, does not have in general a natural choice of volume element.  A
solution is to redefine wave functions to have values in
$L\otimes\sqrt{\wedge^n_\C P}$. The $L^2$-completion of the subspace of
polarized wave functions is a Hilbert space $\mathcal{H}_P$ that represents
quantum states.

The polarization also allows us to define a \emph{representation}
$\mathcal{O}_P\to\End(\mathcal{H}_P)$ of quantizable classical observables into
quantum operators.

The best known example of a real polarization is the vertical polarization
$P=VT^*Q=\ker \tau^\vee_Q$ of the phase space of any \emph{natural mechanical
  system} described by the symplectic manifold $M=T^*Q$
endowed with the standard symplectic form $\omega_0$ or a charged symplectic
form $\omega_0+B$ and with a Riemannian metric $g$ on the configuration space
$Q$. If we consider the natural Darboux coordinates $(q^i,p_i)$ for $\omega_0$,
then the vertical polarization is spanned by the vector fields $\pd{}{p_i}$,
the family of commuting observables is $(q^i)$ and quantizable observables
$f\in\mathcal{O}_P$ are linear functions of momenta $f=X^i(q^j)p_i +
f^0(q^j)$. This is the \emph{Schr\"odinger representation}. Note that $P$ is
naturally isomorphic to $T^*Q\times_Q T^*Q$.

One of the key problems that appear in geometric quantization is the need to
quantize observables whose Hamiltonian vector fields do not respect the chosen
polarization. The most important and basic example is the kinetic energy
$K=\frac{1}{2m}\sum_i p_i^2$ of a natural mechanical system with respect to
the Schr\"odinger representation. The standard method for quantizing such
observables is the Blattner-Kostant-Sternberg (BKS) method \cite{Sni80,Woo92}.
The quantum operator corresponding to ``bad'' observables is defined as a kind
of Lie derivative of the wave functions which are first dragged along the flow
of the corresponding Hamiltonian vector field and then projected to the initial
Hilbert space through the BKS kernel. Recent research \cite{FMMN06} has shown
that the BKS procedure can be interpreted as parallel transport with
respect to a connection on an infinite--dimensional space, and is related to
the Coherent State Transform \cite{Ha02}. BKS also plays an important role in
studying the Maslov correction to semiclassical states \cite{EMN15}.

The remark by which we started our investigation is that \emph{real reducible
  polarizations may have an additional geometric structure}. Our main example
is the vertical polarization $P=T^*Q\times_M T^*Q$ of a natural mechanical
system $(T^*Q,\omega_0)$, where $(Q,g)$ is a Riemannian manifold. Clearly, $P$
can be endowed with the connection obtained by pulling-back the Levi-Civita
connection $\nabla^g$ of the Riemannian manifold $(Q,g)$. Strangely enough,
this connection has never been considered in GQ.

Therefore, we decided to explore the consequences implied by using such a
connection on the polarization bundle. With this aim in mind we had to
understand the mathematical hypotheses under which the introduction of a new
connection would be most natural.  A derivation along $P$ is needed in GQ in
order to define the Hilbert space of states, \emph{i.e.} the polarized wave
functions.  When the space of leaves $M/P$ is a manifold it may happen that
both the line bundle $\mathcal L$ and the polarization $P$ descend to $M/P$, as
in the case of natural mechanical systems, and the action of quantum operators
on polarized sections is given through the Lie derivative with respect to
fundamental horizontal vector fields on the factor $\mathcal L$ and through the
ordinary Lie derivative on the factor $\sqrt{\wedge^n_\C P}$.

It turns out that he most natural framework for considering alternative
connections on the polarization bundle $P$ is that of \emph{real reducible
  polarizations} (Section~\ref{sec:geom-polar-modif}). A reducible polarization
$P$ on a symplectic manifold $M$ is by definition isomorphic to the vertical
bundle of a surjective submersion $\pi\colon M\to Q$. One also needs
reducibility hypotheses that guarantee the descent of $\mathcal L$ and, to some
extent, of the connection $\nabla^{\mathcal L}$. The reducibility of the
prequantum structure $(\mathcal L,\langle\ ,\, \rangle^{\mathcal
  L},\nabla^{\mathcal L})$ is defined as the existence of a complex line bundle
${L}\to Q$ with Hermitian product $\langle\ ,\, \rangle$ such that $\pi^*{L}$
is isomorphic to $\mathcal L$ and $\pi^*\langle\ ,\, \rangle$ is isomorphic to
$\langle\ ,\, \rangle^{\mathcal L}$. If the reduced line bundle ${L}\to Q$
admits a Hermitian connection $\nabla$, then the compatibility condition that
we require between $\nabla^{\mathcal L}$ and ${\nabla}$ is that
$\nabla^{\mathcal L}$ and $\pi^*\nabla$ must define the same space of polarized
sections.

Using the above definition of reducibility we are able to prove one of the main
results of our paper.

\textbf{Theorem.} {\itshape A quantum structure $(\mathcal L,\langle\ ,\,
  \rangle^{\mathcal L},\nabla^{\mathcal L},P)$ is reducible over a submersion
  $\pi\colon M\to Q$ if and only if $\mathcal L$ admits a relative connection
  $\widehat{\nabla}$ along the fibres of $\pi$ that has trivial holonomy
  groups. In this case $P$ turns out to be isomorphic to $\pi^*T^*Q$. If there
  exists a reduction $(L,\langle\ ,\, \rangle, {\nabla})$ of the quantum
  structure such that $\nabla$ is compatible with $\nabla^{\mathcal L}$ then
  $M$ is endowed with the structure of a Liouville manifold \cite{Va}.}

The above Theorem leads to interesting consequences at the quantum level. In
this paper we develop the Geometric Quantization of Liouville manifolds and we
give a complete description of quantizable observables
(Theorem~\ref{Theo:charact-quantizable-observ}): they are functions which are
linear with respect to the fibre coordinates of $\pi$. Such a definition makes
is a generalization to Liouville manifolds of what was already known for
natural mechanical systems.

We prove that polarized sections have the form $\psi\colon Q\to {L}\otimes
K^{1/2}_Q$, where $K^{1/2}_Q$ is the square root of the determinant bundle of
$T^*Q\to Q$. Then, we define quantum operators acting on polarized sections
using a more general idea than the usual GQ approach. Namely, instead of
operating on $K^{1/2}_Q$ using Lie derivatives, we allow for an arbitrary
linear connection $\nabla^Q$. The consequences of this definition are quite
interesting in our opinion.

First of all, we have to check that the modified quantum operator satisfy the
usual properties dictated by Physics. We prove the following

\textbf{Theorem~\ref{sec:quant-class-observ}}
\begin{itemize}{\itshape
\item the commutation identity $\widehat{[f_1,f_2]} = 
i\hbar[\hat f_1, \hat f_2]$
  involving quantizable observables $f_i$ and the corresponding quantum
  operators $\hat f_i$ holds if and only if $\nabla^Q$ is a flat
  connection;
\item the quantum operator $\hat f$ is symmetric if and only if the gradient of
  $f$ with respect to the fibres of $\pi$ is a divergence-free vector field.}
\end{itemize}
The flatness of $\nabla^Q$ is really a very mild restriction since the
determinant bundle $K^{1/2}_Q$ is trivial when $Q$ is orientable. In
particular, when quantizing natural mechanical systems using the vertical
polarization it is possible to act on half-forms via $\nabla^g$ instead of the
Lie derivative. Indeed, $\nabla^g$ is flat on the determinant bundle. The
divergence-free condition is not really a significant restriction since in the
main physical examples, like the quantization of positions, momenta or the
components of the angular momentum, it always holds. The only change that we
will notice at this stage with respect to GQ is the fact that the divergence
terms in the quantum momentum operators will disappear.

When quantizing observables that do not preserve the polarization via the
BKS method, one might use the parallel transport with respect to the arbitrary
connection $\nabla^Q$ instead of the flow of the Hamiltonian vector field.
In natural mechanical systems with the vertical polarization, the usual quantum
energy operator obtained by applying the standard rules of GQ includes a
multiplicative term of the form $\hbar^2 k r_g$, where $r_g$ is the scalar
curvature of $g$, and $k$ is a constant. If we act on half-forms via $\nabla^g$
instead of the Lie derivative, we have an interesting consequence.

\textbf{Theorem~\ref{sec:modif-geom-quant}.}  {\itshape The quantization of a
  particle moving in a Riemannian manifold using the above modified GQ
  procedure acting on half-forms by parallel transport with respect to
  the Levi-Civita connection $\nabla^g$ yields the quantum energy operator
  \begin{displaymath}
    \hat{H} = - \frac{\hbar^2}{2}\Delta + V.
  \end{displaymath}
Therefore, the constant term in front of the scalar curvature vanishes, $k=0$.}

The scalar curvature term appeared for the first time in \cite{dew} with value
$k=1/12$ in the context of canonical quantization, which is GQ in Darboux
coordinates. The author also stated that he got the same result with Feynman
path integral; the proof was published later \cite{Che} (with $k=1/6$) (see
also~\cite{BAL73}, with $k=-1/12$ and \cite{Mos94} with $k=1/8$, although in a
slightly different context). Weyl quantization also shows a degree of ambiguity
with $k=1/12$, $k=1/8$ or $k=1/4$ depending on ordering problems in \cite{LQ92}
(see also \cite{Lan98}). In a recent paper \cite{BAHRR} the Darboux III
oscillator is considered, and it is found that maximal superintegrability is
achieved if the Laplacian is supplemented with the conformal term
$(\hbar^2(n-2)/(8(n-1)))r_g$, where $n=\dim Q$. In GQ the factor is $k=1/12$
\cite{Sni80,Woo92,Wu}.

The ambiguity in the factor $k$ in front of $r_g$ was already noted by De Witt:
\begin{quote}
  The choice of a numerical factor to stand in front of $\hbar^2r_g$ is
  undetermined; all choices lead to the same classical theory in the limit
  $\hbar\to 0$. This, however, is the only ambiguity in the quantum
  Hamiltonian. \cite[p.\ 395]{dew}.
\end{quote}

Another geometric approach to quantum mechanics, namely Covariant Quantum
Mechanics \cite{CanJadMod95,JadMod92}, provides further insight into the
problem. The approach aims at a covariant formulation of classical and quantum
time-dependent particle mechanics by means of a Lagrangian formulation, with
the vertical polarization as a `fixed' representation. In that approach one
uses half-form valued wave functions and quantum operators act on half-forms
through the Levi-Civita covariant derivative. The Schr\"odinger operator is
derived through a Lagrangian approach. The scalar curvature appears in one of
the summands of the Lagrangian. It was proved in \cite{JanMod02} that there is
a unique covariant Schr\"odinger operator up to the multiplicative term
$\hbar^2 k r_g$ in which $k$ cannot be determined by covariance arguments. This
is obviously compatible with our results.

After the above considerations it is natural to ask if there is any evidence
for the scalar curvature term in the energy spectrum of a physical system. The
answer is that to our best knowledge the only examples with a nonzero scalar
curvature have \emph{constant} scalar curvature. Here we mention the rigid
body, either free or subject to a rotationally invariant magnetic field
\cite{ModTejVit08a,ModTejVit08b,Tej01}, and the Landau problem leading to the
Hall effect \cite{lopez3,LopTej01,Tej06-1,Tej06,KMMW}. Needless to say, the
effect of constant scalar curvature on the energy spectrum is a non-measurable
overall constant shift. Indeed, it is quite a difficult task to find a
classical system which would pass all the requirements and topological
obstructions for GQ and still has a nonzero and nonconstant scalar curvature.

\emph{It is our opinion that the only way to decide weather the scalar curvature
  is present or not in the energy operator, and eventually its numerical
  factor, is by means of an experiment.}

It is at the moment unclear if such an experiment can be performed on a
microscopic scale or if any evidence can be inferred from \emph{eg} cosmic
observations. However, we stress that a major argument in favour of $k=0$ (or
no scalar curvature term in the quantum energy operator) and eventually our
modified geometric quantization is the following one.

As the dynamics of a free particle is only ruled by the metric and the topology
of the Riemannian manifold $(Q,g)$ the classical trajectories are
geodesics. The natural way to lift trajectories to the tangent and cotangent
bundle is parallel transport. \emph{So, it is more likely, in our opinion, to
  expect that the quantization procedure should be developed using parallel
  transport of volumes rather than Lie derivative with respect to Hamiltonian
  vector fields.}

The plan of the paper is the following. In section 2, after some mathematical
preliminaries on the theory of connections, we give a summary of those aspects
of GQ which will be analyzed later. Section 3 is devoted to a detailed
description of the modified GQ. In section 4, we consider systems whose
configuration space is a Riemannian manifold and apply to them the modified GQ
method in the Schr\"odinger representation. A final discussion with future
research perspectives will close the paper.

All differentiable manifolds considered in this paper are assumed to be
Hausdorff and paracompact and maps between them are $\cin$. Submersions are
always assumed to be surjective.

\section{Preliminaries}

\subsection{Connections and half-forms}
\label{sec:math-prel}

We start by recalling some mathematical results that will be used throughout
the paper. Given a manifold $M$ we denote by $\tau_M\colon TM\to M$, the
tangent projection and by $\tau^\vee_M\colon T^*M\to M$ the cotangent
projection.

If $p\colon P\to M$ is a vector bundle and $\nabla$ is a linear connection on
it, then there is a natural splitting $TP=H^\nabla(P)\oplus VP$, where $VP=\ker
Tp$ is the \emph{vertical bundle}, or the space of vectors tangent to the
fibres of $p$ and $H^\nabla(P)$ is the \emph{horizontal bundle}, whose fibres
are isomorphic to the fibres of $TM$. Given a vector field $X\colon P\to TP$,
the connection allows us to compute its components with respect to the above
direct sum splitting. If $(x^\alpha)$ are coordinates on $M$, $(b_k)$ is a
local basis of $P$ with dual coordinates $( y^k)$,
$\pd{}{x^\alpha}=\partial_\alpha$ is a local basis of $TM$ such that
$\nabla_{\partial_\alpha}b_j = \Gamma^k_{\alpha j}b_k$ and
$X=X^\alpha\partial_\alpha + X^kb_k$, then the coordinate expression of the
horizontal part $h^\nabla(X)$, or \emph{horizontalization of $X$} is
\begin{equation}
  \label{eq:3}
  h^\nabla(X) = X^\alpha(\partial_\alpha - \Gamma^k_{\alpha j} y^j b_k),
\end{equation}
where $X^\alpha\in\mathcal{C}^\infty(P)$. One can also lift a vector field
$X\colon M\to TM$ to a horizontal vector field $h^\nabla( X)$ that projects
onto $X$. The coordinate expression will be the same, this time with
$X^\alpha\in\mathcal{C}^\infty(M)$. This theory is well-established in
differential geometry and can be found, for instance, in \cite{KolMicSlo93}.

It is well known that we can define the Lie derivative of a section $\xi$ of
the vector bundle $p\colon P\to M$ with respect to any projectable vector field
$X$ on $P$ and this gives a new section $L_X \xi$ of $P$. If $X$ is a vector
field on $M$ that has a natural lift $X^\ell$ to a vector field on $P$, then we
denote simply by $L_X\xi$ the Lie derivative of $\xi$ with respect to
$X^\ell$. There is a natural relationship between covariant and Lie derivatives
induced by horizontalization. More precisely, if $X\colon P\to TP$ is a vector
field that projects onto a vector field ${X} \colon M\to TM$, then one
has \cite[p.\ 376]{KolMicSlo93}
\begin{equation}
  \label{eq:10}
  L_{h^\nabla(X)}\xi = \nabla_{{X}} \xi.
\end{equation}
In particular, if we start from the horizontalization of a vector field on the
base space ${X}\colon M\to TM$, then we obtain
$L_{h^\nabla({X})}\xi = \nabla_{{X}} \xi$.

It is well known that a complex line bundle $K\to M$ whose first Chern class
reduced modulo 2 vanishes, $0=[c_1(K)]\in H^2(M,\Z/2)$, admits a square root
bundle $\sqrt{K}\to M$ which is not necessarily unique in general.  The Lie
derivative of a section of this bundle with respect to a vector field $X$ on
$M$ that has a natural lift $X^\ell$ to $K$ can be easily defined by requiring
that Leibniz's rule holds. That is, if $\sqrt{\nu}\colon M\to K$ is a section
such that $\sqrt{\nu}\otimes\sqrt{\nu} = \nu\colon M\to K$ then
\begin{equation}
  \label{eq:7}
  L_X\sqrt{\nu}\otimes \sqrt{\nu} + \sqrt{\nu}\otimes L_X\sqrt{\nu}
  = L_X\nu.
\end{equation}
In coordinates if $\sqrt{\nu}=f\sqrt{b}$, where $f\in\mathcal{C}^\infty(M)$ and
$\sqrt{b}$ is a local basis of $\sqrt{K}$, we have
\begin{equation}
  \label{eq:6}
  L_X\sqrt{\nu} = \left(L_Xf + \frac{1}{2}(L_Xb)_0\right)\sqrt{b},
\end{equation}
where $(L_Xb)_0\in\mathcal{C}^\infty(M)$ is defined by $L_Xb = (L_Xb)_0b$.  For
instance, if $M$ is oriented, endowed with a Riemannian metric $g$ and if
$K=\wedge^n_\C T^*M$ is the complexification of the bundle of volume forms,
then there exists a unique trivial square root bundle $\sqrt{\wedge^n_\C
  T^*M}\to M$ whose sections are called \emph{half-forms}. There exists a
distinguished nowhere-vanishing half-form $\sqrt{\nu_g}$ whose square is the
oriented Riemannian volume element $\nu_g$. Since every vector field $X$ on $M$
has a natural lift to the bundle of volume elements, one can define the Lie
derivative of any half-form with respect to it and one gets
$L_X\sqrt{\nu_g}=\frac{1}{2}\dive_g(X)\sqrt{\nu_g}$, where $\dive_g(X)$ is the
divergence of $X$ with respect to the Riemannian metric $g$.  If $(x^\alpha)$
are oriented coordinates on $M$ then $\sqrt{\nu_g} =
\sqrt[4]{|g|}\sqrt{dx^1\wedge\cdots\wedge dx^n}$ and
\begin{equation}
  \label{eq:12}
  L_X\sqrt{\nu_g} =
  \frac{1}{2}\frac{\partial_\alpha(X^\alpha\sqrt{|g|})}{\sqrt{|g|}}\sqrt{\nu_g}.
\end{equation}

The same procedure can be used to define a linear connection in $\sqrt{K}\to M$
from a linear connection $\nabla$ in $K\to M$. Indeed,
\begin{equation}
  \label{eq:13}
  \nabla_X\sqrt{\nu}\otimes \sqrt{\nu} +
  \sqrt{\nu}\otimes \nabla_X\sqrt{\nu} = \nabla_X\nu.
\end{equation}
This implies $\nabla_X\sqrt{\nu} = \left(L_Xf +
  \frac{1}{2}(\nabla_Xb)_0\right)\sqrt{b}$ with an obvious meaning of symbols.
In particular the Levi-Civita connection $\nabla$ on the bundle $\wedge^n_\C
T^*M$ of volume elements induces a connection on the bundle of half-forms. If
we take $b=\sqrt{dx^1\wedge\cdots\wedge dx^n}$ then we have
\begin{equation}
  \label{eq:14}
  \nabla_X\sqrt{\nu} = \left(L_Xf -
    \frac{1}{2}X^\alpha\Gamma^\beta_{\alpha \beta}f\right)
  \sqrt{dx^1\wedge\cdots\wedge dx^n}.
\end{equation}

\begin{Remark}\label{sec:conn-half-forms}
  Since the oriented Riemannian volume element $\nu_g$ is parallel for the
  Levi-Civita connection, $\nabla\nu_g=0$, it follows that we also have
  $\nabla\sqrt{\nu_g} = 0$. This means that both $\wedge^n_\C T^*M$, endowed
  with the Levi-Civita connection, and its square root $\sqrt{\wedge^n_\C
    T^*M}$, with the induced connection defined above, are flat line bundles.
\end{Remark}

We recall that if $F\to Q$ is a vector bundle and $\pi\colon M\to Q$ is a
submersion then, given a connection $\olnabla$ on $F\to Q$, the pull-back
connection $\pi^*\olnabla$ is defined by
\begin{equation}
  \pi^*\olnabla_X\pi^*s = \pi^*(\olnabla_{\bar X}s)\label{eq:26}
\end{equation}
where $s\colon Q \to F$ is a section of $F$ and $X\colon M\to TM$ is a vector
field that projects onto the vector field ${X}\colon Q\to TQ$. The Christoffel
symbols of $\pi^*\olnabla$ in the direction of $VM=\ker T\pi$ are zero.

On the other hand the pull-back bundle $\pi^*F\to M$ is endowed with a natural
\emph{partial connection}, or \emph{$\pi$-relative connection}
$\widehat\nabla^\pi$ in the direction of $VM \subset TM$, which is
defined as $\widehat\nabla^\pi_V(f\,\pi^*s)=Vf\cdot \pi^*s$ for any $f\in
\cin(M)$, $s\in\Gamma(Q,F)$ and any vertical vector field $V$. Obviously this
partial connection is flat.

Any connection $\nabla$ on a vector bundle $p\colon E\to M$ induces by
restriction to the vertical bundle $V\pi$ a $\pi$-relative connection that we
denote $\pi_{M/Q}(\nabla)$. A straightforward computation proves the following
result.

\begin{Lemma}\label{Lem:canonical-relative-connection-on pullbacks} Let
  $\pi\colon M\to Q$ be a submersion and let $q\colon F\to Q$ be a vector
  bundle. For any connection $\overline\nabla$ on $F$, the pullback connection
  $\pi^*(\overline\nabla)$ on the pullback vector bundle $\pi^* F\to M$ induces
  the natural $\pi$-relative connection $\widehat\nabla^\pi$ on $\pi^* F$. That
  is
  \[
  \pi_{M/Q}(\pi^*(\overline\nabla))=\widehat\nabla^\pi.
  \]
\end{Lemma}

\subsection{Geometric quantization}
\label{sec:prel-geom-quant}

In this section we recall the basic ingredients of GQ.  Our main sources
are~\cite{Sni80,Woo92}. We will focus on the aspects that will play a key role
for the results of the present paper. A \emph{prequantum structure} on a
symplectic manifold $(M,\omega)$ is a triple $\mathcal Q=(\mathcal L,\langle\,
, \rangle^{\mathcal L},\nabla^{\mathcal L})$, where $ \mathcal L\to M$ is a
complex line bundle endowed with a Hermitian metric $\langle\, ,
\rangle^{\mathcal L}$ and a Hermitian connection $\nabla^{\mathcal L}$ such
that
\begin{displaymath}
  R[\nabla^{\mathcal L}]=-i\frac{\omega}{\hbar}\Id_{\mathcal L},
\end{displaymath}
where $R[\nabla^{\mathcal L}]$ is the curvature of $\nabla^{\mathcal L}$ and
$\hbar=\frac{h}{2\pi}$ is the reduced Planck constant.  It can be proved that
the symplectic manifold $(M,\omega)$ admits a prequantum structure if and only
if the cohomology class $[\frac{\omega}{h}]$ is integral, that is
\begin{displaymath}
  [\frac{\omega}{h}]\in i(H^{2}(M,{\mathbb{Z}}))\subset
  H^{2}(M,{\mathbb{R}}),
\end{displaymath}
where $i$ is the map induced in cohomology by the inclusion
$\mathbb{Z}\hookrightarrow \mathbb{R}$. If the above condition is fulfilled,
then $(M,\omega)$ is said to be \emph{quantizable}.

One says that two Hermitian line bundles $(L, {\langle\, , \rangle})$,
$(L^\prime, {\langle\, , \rangle}^\prime)$ (Hermitian line bundles with
connection $(L, {\langle\, , \rangle}, \nabla)$, $(L^\prime, {\langle\, ,
  \rangle}^\prime, \nabla^\prime)$) are equivalent Hermitian bundles (with
connection) if there exists an isomorphism of line bundles $\phi\colon L\to
L^\prime$ such that $\phi^*{\langle\, , \rangle}^\prime={\langle\, , \rangle}$
(and $\phi^*\nabla'=\nabla$).  The set of equivalence classes of prequantum
structures is parametrized by $H^{1}(M,U(1))\simeq \Hom(\pi_1(M),U(1))$.
Therefore, a symplectic manifold $M$ can admit non-equivalent
quantizations only if it is not simply connected.

\begin{Definition}
  We say that a prequantization structure $(\mathcal L, {\langle\, ,
    \rangle}^{\mathcal L}, \nabla^{\mathcal L})$ on $(M,\omega)$ is
  \emph{reducible} over a submersion $\pi\colon M\to Q$ if there exists a
  Hermitian line bundle $(L,\nabla)$ on $Q$ such that $(\mathcal L, {\langle\,
    , \rangle}^{\mathcal L}) $ is equivalent to $(\pi^*{L}, \pi^*\langle\, ,
  \rangle{})$. In this case we say that $(L, {\langle\, , \rangle}, \nabla)$ is
  \emph{a reduction} of the prequantization structure.\end{Definition} Let us
stress that $\nabla^{\mathcal L}$ cannot be equivalent to $\pi^*{\nabla}$, for
a Hermitian connection $\nabla$ on ${L}\to Q$, since $i\hbar \pi^*R[{\nabla}]$
is not a symplectic form because the vertical vector fields of the submersion
$\pi\colon M\to Q$ belong to its radical. We have the following result.

\begin{Proposition}\label{Prop:red-quant-struct} If a prequantization structure
  $(\mathcal L, {\langle\, , \rangle}^{\mathcal L}, \nabla^{\mathcal L})$ on
  $(M,\omega)$ is reducible over a submersion $\pi\colon M\to Q$ with reduction
  $( L, \langle\, , \rangle{}, \nabla)$, then there exists a $1$-form
  $\alpha\in\Omega^1(M)$ such that
  \begin{equation}
    \omega=\pi^*\bar\omega+d\alpha,\label{eq:31}
  \end{equation}
  where $\bar\omega=i\hbar R[\nabla]\in\Omega^2(Q)$. Moreover, $d\alpha$ is a
  symplectic structure on $M$.
\end{Proposition}

\begin{proof} Since $\mathcal L\simeq\pi^*L$ one has
  $\nabla^{\mathcal L}-\pi^*\nabla=-\frac{i}{\hbar}\alpha$ for a certain $1$-form
  $\alpha\in\Omega(M)$. The first claim follows by computing the curvatures of
  these connections bearing in mind the quantization condition. Since $d\alpha
  =\omega-\pi^*\bar\omega$, and $\omega$ is symplectic, it follows immediately
  that $d\alpha$ is a symplectic form.
\end{proof}

Note that when a symplectic manifold $(M,\omega)$ admits a prequantization
structure that is reducible over a submersion $\pi\colon M\to Q$, then $M$ is
not compact. Indeed, by Proposition \ref{Prop:red-quant-struct}, there exists a
$1$-form $\alpha\in\Omega^1(M)$ such that $(M,d\alpha)$ is a symplectic
manifold, and it is well known that an exact symplectic manifold is always not
compact.

One has the following descent result, whose proof can be found in \cite{tv16}.

\begin{Theorem}\label{Theo:descent-absolute-parallelism} Let $\pi\colon M\to Q$
  be a surjective submersion and consider a vector bundle $p\colon E\to M$ that
  admits an absolute parallelism $\widehat\nabla$ relative to $\pi$. Then,
  there exists a vector bundle $q\colon F\to Q$, unique up to isomorphism over
  $N$, and a vector bundle isomorphism $\varphi\colon
  E\xrightarrow{\sim}\pi^*F$ that transforms $\widehat\nabla$ into the natural
  flat relative connection $\widehat\nabla^\pi$ defined on the pullback vector
  bundle $\pi^*F\to M$.
\end{Theorem}

This together with a straightforward computation, see \cite{tv16}, gives the
following result

\begin{Corollary}\label{Cor:red-quantum-structure} A prequantization structure
  $(\mathcal L, {\langle\, , \rangle}^{\mathcal L}, \nabla^{\mathcal L})$ on a
  symplectic manifold $(M,\omega)$ is reducible over a submersion $\pi\colon
  M\to Q$ if and only if $\mathcal L$ admits a $\pi$-relative connection
  $\widehat\nabla$ which is an absolute parallelism relative to $\pi$; that is,
  if and only if the holonomy groups of $\widehat\nabla$ along the fibers of
  $\pi$ are trivial.
\end{Corollary}

An example is provided by the cotangent bundle $\tau^\vee_Q\colon T^*Q \to Q$
of a manifold $Q$. Taking into account Corollary
\ref{Cor:red-quantum-structure}, the topological triviality of the fibers of
$\tau^\vee_Q$ implies that every prequantum structure on the symplectic
manifold $(T^*Q,\omega_0=d\theta)$, where $\theta$ is the Liouville form, is
reducible over $\tau^\vee_Q$.

A \emph{real polarization} on a symplectic manifold $(M,\omega)$ is an
involutive Lagrangian distribution $P\subset TM$. A \emph{complex polarization}
on a symplectic manifold $(M,\omega)$ is an involutive Lagrangian distribution
$P$ of the complexified tangent bundle of $M$, $P\subset T_{\mathbb{C}}M$ that
fulfills additional hypothesis \cite{Sni80,Woo92}.

A polarization is needed in order to define the Hilbert space of quantum states
and also to select the class of observables to be quantized.

A \emph{quantization} of a symplectic mechanical system consists of a
prequantum structure $(\mathcal L,\langle\, , \rangle^{\mathcal L},
\nabla^{\mathcal L})$ and a polarization $P$. After such a choice, a Hilbert
space $\cH_P$ can be defined as follows. Sections which are covariantly
constant in the direction of $P$ represent quantum states. In the particular
case $M=T^*Q$, the simplest real polarization is the vertical one,
$P=VT^*Q\simeq T^*Q\times_Q T^*Q$, and the corresponding quantization is called
the \emph{Schr\"odinger representation}.  Notice however that the scalar
product $(\psi_1,\psi_2) = \int_{M}\langle \psi,\psi\rangle
\operatorname{vol}_{\omega}$ of two parallel sections $\psi_1$, $\psi_2$ might
diverge.

The bundle of \emph{half-forms} $N_P^{1/2}=\sqrt{\wedge^n_\C P}\to M$ was
introduced in order to avoid the above problem. Such bundle does only exist
provided that the square of the first Stiefel-Whitney class of $P$ vanishes,
$w_1(P)^2=0$, \cite{Sni80,Woo92}. This amounts to the possibility of making a
consistent choice for the square root of the transition functions of the line
bundle $\wedge^n_\C P\to M$ to allow for consistent coordinate changes in
$N_P^{1/2}$. This is equivalent to the datum of a \emph{metaplectic structure}
for the symplectic manifold $(M,\omega)$
\cite{Sni80,Woo92}.

In order to be able to achieve our goals, it is of great importance the fact
that $N^{1/2}_P$ admits a partial flat connection $\widehat{\nabla}^B$ in the
direction of $P$; $\widehat{\nabla}^B$ is called the \emph{Bott connection}
(see Subsection~\ref{sec:geom-polar}).  Indeed, $P$ admits a local basis
$\xi_1$, \dots, $\xi_n$ formed by Hamiltonian vector fields. Such fields
commute as they take their values in a Lagrangian subspace.  Moreover, it can
be proved \cite{Sni80} that the existence of a metaplectic structure implies
that any two such local bases are connected by a change of coordinates which is
constant along the polarization. This means that if $\nu=
\nu_0\sqrt{\xi_1\wedge\cdots\wedge\xi_n}$, where
$\nu_0\in\mathcal{C}^\infty(M)$, and $Y\colon M\to P$ then
\begin{equation}
  \label{eq:8}
  \widehat{\nabla}^B_Y\nu = (L_Y\nu_0)\sqrt{\xi_1\wedge\cdots\wedge\xi_n},
\end{equation}
is well defined. It can be verified that $\widehat{\nabla}^B$ is flat
\cite{Sni80}.

One can then consider \emph{polarized sections}, that is sections of the
prequantization line bundle twisted by half-forms, $\mathcal L\otimes
N_P^{1/2}\to M$, which are covariantly constant in the direction of $P$ with
respect to the partial connection on the tensor product obtained from
$\nabla^{\mathcal L}$ and $\widehat{\nabla}^B$. Such sections can be integrated
over the space of leaves of the polarization $M/P$. The \emph{Hilbert space
  $\mathcal{H}_P$ of quantum states} is the $L^2$-completion of the space of
square-integrable sections ${\psi}$ of $\mathcal L\otimes N^{1/2}_P$ which are
covariantly constant along $P$.

The space of (straightforwardly) \emph{quantizable observables}
${\cO}_P\subset\mathcal{C}^{\infty}(M)$ consists of functions
$f\in\mathcal{C}^{\infty}(M)$ whose Hamiltonian vector field $X_f$ preserve the
polarization $P$; in other words, $[X_f,P]\subset P$. A choice of polarization
corresponds to what in the Physics literature is called a choice of
\emph{representation}.

At this point we need to define Dirac's correspondence between the space of
quantizable observables $\mathcal{O}_P$ and quantum operators acting on the
Hilbert space $\mathcal H$. Every classical observable
$f\in\mathcal{C}^\infty(M)$ yields a Hamiltonian vector field $X_f$ on
$M$. Such a vector field can always be lifted to a vector field on the
prequantum line bundle $\mathcal L\to M$ that operates on its sections
$\psi\colon M\to \mathcal L$. However, we also need to operate on the factor
$N^{1/2}_P$ and the Lie derivative can be defined only if we use vector fields
that preserve $P$.

Every quantizable observable $f\in{\cO}_P$ corresponds to a symmetric operator
$\hat f\colon \cH_P\to\cH_P$ and in many important cases one can prove that it
is self-adjoint.  The operator $\hat f$ is defined as follows. The Hamiltonian
vector field corresponding to $f$ can be lifted to a complex Hermitian vector
field $X^\ell_f$ on $\mathcal L$ that preserves the connection
$\nabla^{\mathcal L}$, in the sense that $L_{X^\ell_f}\nabla = 0$, regarding
the connection as a tensor field, see \cite{KolMicSlo93}. The expression of
$X^\ell_f$ is
\begin{equation}
  \label{eq:5}
  X^\ell_f = h^\nabla(X_f) + i\frac{f}{\hbar}E,
\end{equation}
where $E$ is the Euler vector field of $\mathcal L$, which is the vertical
vector field such that for every $z\in \mathcal L$ one has $E_z=(z,z)$ via the
natural isomorphism $V\mathcal L=\mathcal L\times_M \mathcal L$.  The vector
field $X^\ell_f$ acts on sections $s\colon M\to \mathcal L$ through its flow
$\phi^{f,\mathcal L}_{t}$ as $\phi^{f,\mathcal L}_{t,*}(s) = \phi^{f,\mathcal
  L}_{ - t} \circ s \circ \phi^{f}_{t}$, where $\phi^{f}_{t}$ is the flow of
$X_f$. The generalized Lie derivative (see \emph{eg} \cite{KolMicSlo93}) of the
section $s$ is then defined via a limit $t\to 0$; one has \cite[p.\
378]{KolMicSlo93}
\begin{equation}
  \label{eq:4}
  i\hbar L_{X^\ell_f}s = i\hbar\nabla_{X_f}s - fs.
\end{equation}

On decomposable sections ${\psi}=s\otimes\sqrt{\nu}\colon M \to \mathcal
L\otimes N_P^{1/2}$ which are covariantly constant along the direction of $P$
(\emph{ie, wave functions}) we can define the action of \emph{quantum
  operators}:
\begin{equation}
  \label{eq:9}
  \hat f(\psi) = i\hbar L_{\tilde{X}_f}{\psi} =
  i\hbar L_{X^\ell_f}s\otimes\sqrt{\nu} +
  i\hbar s\otimes L_{X_f}\sqrt{\nu}.
\end{equation}
Here $\tilde{X}_f$ is the vector field induced on $\mathcal L\otimes N_P^{1/2}$
by $X^\ell_f$ and $X_f$.  Note that the action of $X_f$ on sections of
$N^{1/2}_P\to M$ is well defined since $X_f$ preserves $P$.

One can prove that the well-known relation between classical and quantum
commutators holds:
\begin{equation}
  \label{eq:27}
  [\hat f_1,\hat f_2] = i\hbar \widehat{\{f_1,f_2\}}.
\end{equation}

We stress that not all physically or mathematically interesting classical
observables can be quantized. For instance, in the case of the Schr\"{o}dinger
representation for a natural mechanical system $f\in \mathcal{O}_P$
if and only if $f$ is linear in the momenta, \cite{Sni80,Woo92}. More
intrinsically, if $f\in \mathcal{O}_P$ then one has $f=(\tau_Q^\vee)^*
h+\theta(X^\vee)$, where $h\in\cin(Q)$ and $X^\vee$ is the cotangent lift of a
vector field $X$ on $Q$. This implies that the Hamiltonian, which is the total
energy function $H=K_g+V$, where $K_g(q^i,p_i)=\frac{1}{2}g^{ij}p_ip_j$ and
$V\in\mathcal{C}^{\infty}(Q)$, is not quantizable.

The quantization of energy is usually achieved in the framework of the BKS
theory of quantization of observables that do not preserve the polarization. If
$H$ is such an observable, we denote by $\phi^H_t$ the flow of $X_H$; we have
$\phi^H_{t,*}(P)\neq P$. We assume that for $0 < t< \epsilon$ the two
polarizations $\phi^H_{t,*}(P)$ and $P$ are \emph{strongly admissible}: this
means that they must be transverse and that some technical hypothesis have to
be verified \cite{Sni80}. Let us denote by $\cH_t=\cH_{\phi^H_{t,*}(P)}$ the
Hilbert space of quantum observables corresponding to the polarization
$\phi^H_{t,*}(P)$, in particular one has $\cH_0=\cH_P$. Then a map
$\cH_{0t}\colon\cH_t \to \cH_0$, the BKS kernel, is defined for every
${\psi}_2\in\cH_t$ by
\begin{equation}
  ({\psi}_1,\cH_{0t}({\psi}_2))=
  \int_{M/P}\langle {\psi}_1,{\psi}_2 \rangle
  \quad\forall{\psi}_1\in\cH_0.
  \label{eq:15}
\end{equation}
The integral on the right-hand side is the BKS pairing, see \cite{Sni80} for
details. The quantum operator associated to $H$ is defined, for
$\psi=s\otimes\sqrt{\nu}\in\cH_0$, as
\begin{equation}
  \label{eq:16}
  \hat{H}(\psi) =
  i\hbar \frac{d}{dt}\cH_{0t}\circ
  \phi^{H,\mathcal L}_{t,*}(s)\otimes\phi^H_{t,*}(\sqrt{\nu})\big|_{t=0},
\end{equation}
In the Schr\"{o}dinger representation of a natural mechanical system,
if $(q^i,p_i)$ are coordinates on $T^*Q$ we have the
following quantum operators:
\begin{gather}
  \label{eq:17}
  \widehat{q^i}(\psi) = i\hbar q^i s\otimes\sqrt{\nu_g}, \\\label{eq:20}
  \widehat{f^ip_i}(\psi) = - i\hbar\left(f^i\pd{}{p_i}\psi^0 +
    \frac{1}{2}\frac{\partial_i(f^i\sqrt{|g|})}{\sqrt{|g|}}\psi^0\right)
  b_0\otimes\sqrt{\nu_g}, \\\label{eq:21} \widehat H= \left(-
    \frac{\hbar^2}{2}\Bigr(\Delta(\psi) - \frac{r_g}{6}\psi\Bigl) +
    V\psi\right)\otimes \sqrt{\nu_g}.
\end{gather}
In the above formulae, $f^i$ is a function on $Q$ and
${\psi}=s\otimes\sqrt{\nu_g}$ denotes an element of the Hilbert space, where
$s=\psi^0b_0$ is a section of $L\to Q$ which is expressed through a local basis
$b_0$. Moreover, $\Delta$ is the Bochner Laplacian of $\nabla$ and $r_g$ is the
scalar curvature of the Riemannian metric $g$. The above formulae can be found,
for example, in \cite{Sni80}.

\section{The geometry of polarizations and the modified
  quantization method}
\label{sec:geom-polar-modif}

In this section we want to show how the existence of a partial flat connection
on a vector bundle is related to infinitesimal and global descent properties of
the vector bundle. This is relevant for geometric quantization since this is
exactly the situation that one has when a prequantum line bundle and a
polarization on a symplectic manifold are given.

\subsection{The geometry of polarizations}
\label{sec:geom-polar}
Let $(M,\omega)$ be a symplectic manifold. For any real polarization $P$ of
$(M,\omega)$ we denote by $N(P)=TM/P$ the normal bundle of the polarization and
$\pi^N\colon TM\to N(P)$ is the natural projection.  It is well known that
$N(P)$ has a natural partial connection relative to $P$, we simply say a
$P$-relative connection, defined by
\[
\widehat\nabla^B_V(\pi^N(D))=\pi^N([V,D]),\quad \textrm{for any}\
V\in\Gamma(M,P), D\in \Gamma(M,TM).
\]
The $P$-relative connection $\widehat\nabla^B$ is called the Bott connection of
$(M,P)$ and one easily checks that it is flat. This connection induces a
connection on $N(P)^*$; the symplectic form gives an isomorphism
$P\xrightarrow{\sim} N(P)^*$ by which we define a flat $P$-relative connection
on the vector bundle $P\to M$ that we also denote by $\widehat\nabla^B$.

We recall a definition which is of key importance to our aims.
\begin{Definition}
  A real polarization $P\hookrightarrow TM$ is called \emph{reducible} if there
  exists a submersion $\pi\colon M\to Q$ such that $P$ coincides with the
  vertical bundle $V(M/Q)\to M$ of the submersion $\pi$.
\end{Definition}

Hence $\pi\colon M\to Q$ is a Lagrangian submersion and there is an equivalence
between reducible real polarizations on $M$ and Lagrangian submersions whose
total space is $M$.

If $P$ is reducible over a submersion $\pi$, then $P$-relative connections are
$\pi$-relative connections (see Subsection~\ref{sec:math-prel}).

\begin{Proposition}\label{pro:rel-conn-ind-red-polarizations}
  Let $(M,\omega)$ be a symplectic manifold and let $P\to M$ be a reducible
  real polarization with associated submersion $\pi\colon M\to Q$. Then
  \begin{enumerate}
  \item the symplectic form yields an isomorphism $P\xrightarrow{\sim}
    \pi^*T^*Q$ that maps $\widehat\nabla^B$ into $\widehat\nabla^\pi$;
  \item The relative connection induced on $P\to M$ by the pullback of any
    connection $\nabla^Q$ on the cotangent bundle $T^*Q\to Q$ coincides with
    the flat $\pi$-relative connection $\widehat\nabla^\pi$ induced by the Bott
    connection, that is
  \[
  \pi_{M/Q}(\pi^*(\nabla^Q))=\widehat\nabla^\pi.
  \]
  \end{enumerate}
\end{Proposition}

For any real polarization $P\hookrightarrow TM$ on a symplectic manifold of
dimension $2n$ the bundle of $P$-transversal volume forms is the complex line
bundle $\Lambda_\C^nN(P)^*\to M$. The symplectic form yields an isomorphism of
complex line bundles $N_P=\Lambda_\C^n P\xrightarrow{\sim} \Lambda_\C^nN(P)^*$,
so that we have a flat $P$-relative connection on the line bundle $N_P\to
M$. One says that the polarization $P$ admits a metalinear structure if the
complex line bundle $N_P\to M$ admits a square root $N_P^{1/2}\to M$, which is
called the \emph{bundle of half-forms} transverse to $P$; it is endowed by a
flat $P$-relative connection $\widehat\nabla^B$.

If we consider the Hermitian connection $\nabla^{\mathcal L}$ of a quantum
structure for the symplectic manifold $(M,\omega)$, then the associated
$P$-relative connection $\widehat\nabla^{\mathcal L}=\pi_{M/P}(\nabla^{\mathcal
  L})$ is flat. Indeed one has
\begin{displaymath}
R[{\widehat\nabla}]=\pi_{M/P}(R[\nabla^{\mathcal L}])=
\pi_{M/P}(-i\frac{\omega}{\hbar}\otimes\Id_{\mathcal L})=
-i\frac{\pi_{M/P}(\omega)}{\hbar}\otimes\Id_{\mathcal L}=0,
\end{displaymath}
since $P$ is Lagrangian.

Therefore, the line bundle $\mathcal L\otimes N_P^{1/2}\to M$ can be endowed
with the twisted flat $P$-relative connection $\widehat
\nabla^P=\widehat\nabla^{\mathcal L}\otimes 1+1\otimes \widehat\nabla^B$, where
$\widehat \nabla$ is the $P$-relative connection induced by $\nabla^{\mathcal
  L}$. The space $\Gamma_P(M,\mathcal L\otimes N_P^{1/2})$ of polarized
sections of $\mathcal L$ is the space of sections of the line bundle $\mathcal
L\otimes N_P^{1/2}\to M$ that are covariantly constant for the $P$-relative
connection $\widehat\nabla^P$:
\begin{displaymath}
  \Gamma_P(M,\mathcal L\otimes
  N_P^{1/2})=\{s\in\Gamma(M, \mathcal L\otimes N_P^{1/2})
    \ \colon \widehat\nabla^P_V s=0,
  \text{for every}\ V\in\Gamma(M,P)\}.
\end{displaymath}

It is well known that a complex line bundle $K\to M$ admits a square root if,
and only if, its first Chern class $c_1(K)\in H^2(M;\Z)$ is even.  If $P$ is a
real polarization which is reducible over the submersion $\pi\colon M\to Q$,
then the symplectic isomorphism $P\to \pi^*T^*Q$ implies that $N_P=\pi^*{K_Q}$,
where $K_Q=\Lambda_\C^n T^*Q\to Q$ is the determinant of the complexified
cotangent bundle of $Q$.  Hence $P$ admits a metalinear structure if and only
if $K_Q$ admits a square root $K_Q^{1/2}\to Q$. In particular, if $Q$ is
orientable it follows that $P$ admits a metalinear structure.

In the reducible case, if $\pi^*K_{Q}^{1/2}\to M$ is a metalinear structure,
then any connection $\nabla^Q$ on $K_{Q}^{1/2}\to Q$ endows the line bundle
$\mathcal L\otimes \pi^*K_{Q}^{1/2}\to M$ with the connection
$\widetilde\nabla=\nabla^{\mathcal L}\otimes 1+1\otimes
\pi^*(\nabla^Q)$. Taking into account Proposition
\ref{pro:rel-conn-ind-red-polarizations} we immediately get the following key
result.

\begin{Theorem}\label{sec:geom-polar-2}
  Let $P$ be a polarization which is reducible over the submersion $\pi\colon
  M\to Q$ and let $K^{1/2}_Q\to Q$ be a metalinear structure. Then, every
  linear connection $\nabla^Q$ on $K^{1/2}_Q\to Q$ yields the same space of
  polarized sections
  \begin{equation}
    \label{eq:1}
    \Gamma_P(M,\mathcal L\otimes N_P^{1/2})=\{s\in\Gamma(M,
    \mathcal L\otimes N_P^{1/2})\ \colon \widetilde\nabla_V s=0,\forall\
    V\in\Gamma(M,P)\}.
  \end{equation}
\end{Theorem}
\begin{proof}
  For any connection $\nabla^Q$ on $K_Q^{1/2}\to Q$, the $\pi$-relative
  connection on $\mathcal L\otimes \pi^*K_Q^{1/2}\to M$ induced by the twisted
  connection $\nabla^{\mathcal L}\otimes 1+1\otimes \pi^*(\nabla^Q)$ coincides
  with the flat $\pi$-relative connection $\widehat
  \nabla^P=\widehat\nabla^{\mathcal L}\otimes 1+1\otimes \widehat\nabla^B$.
\end{proof}

\begin{Corollary}\label{sec:geom-polar-1}
  If, in addition to the above hypotheses, the prequantum structure $(\mathcal
  L, {\langle\, , \rangle}^{\mathcal L}, \nabla^{\mathcal L})$ is reducible
  over the submersion $\pi\colon M\to Q$ and $( L, {\langle\, , \rangle},
  \nabla)$ is a reduction, then the space of polarized sections turns out to be
  \begin{equation}
    \label{eq:2}
    \Gamma_P(M,\mathcal L\otimes N_P^{1/2})= \Gamma(Q,{L}\otimes_Q K^{1/2}_Q).
  \end{equation}
\end{Corollary}
\begin{proof}
  Since the line bundle $\mathcal L\otimes_M N_P^{1/2}=
  \pi^*({L}\otimes_QK_Q^{1/2})$ is a pullback, its space of global
  sections is given by
  \begin{displaymath}
    \Gamma(M,\mathcal L\otimes_M N_P^{1/2})=
    \mathcal{C}^\infty(M)\otimes_{\mathcal{C}^\infty(Q)}
    \Gamma(Q,{L}\otimes_QK_Q^{1/2}).
  \end{displaymath}
  Taking into account now that both partial covariant derivatives
  $\nabla^{\mathcal L}$ and $\pi^*({\nabla})$ are a pullback, we immediately
  obtain the required identification.
\end{proof}
\begin{Remark}
  As every connection $\widetilde\nabla$ in the family of connections
  $\widetilde\nabla=\nabla^{\mathcal L}\otimes 1+1\otimes \pi^*(\nabla^Q)$
  parameterized by the connections $\nabla^Q$ of $K_Q^{1/2}$, yields the same
  space of polarized sections $\Gamma_P(M,\mathcal L\otimes N_P^{1/2})$, we can
  quantize observables by using the connection $\widetilde\nabla$ instead of
  the $\pi$-relative connection $\widehat \nabla^P$. While the space of
  polarized sections does not change, we will see that the quantization of
  classical observables changes in a substantial way.
\end{Remark}

\subsection{The modified quantization method}
\label{sec:modif-quant-meth}

In this Subsection we propose a modification of the standard GQ scheme which
will make a fundamental use of the geometric structure that is present on the
space of polarized sections.  More precisely, we will show that while we will
use the same Hilbert space as GQ, we will determine quantizable observables in
the more general situation of a reducible quantum structure.

Then, we make the key remark that the representation of classical observables
as quantum operators on that space \emph{is not unique}. More precisely, the
representation can be modified with respect to the standard prescription of GQ
(see Subsection~\ref{sec:prel-geom-quant}) by choosing any connection
$\nabla^Q$ on the bundle of half-forms $K_Q^{1/2}\to Q$. Indeed, a connection
on $K_Q^{1/2}\to Q$ is usually defined in an natural way from a linear
connection on $\tau^\vee_Q\colon T^*Q\to Q$, so we will assume that a linear
connection $\nabla^Q$ on $\tau^\vee_Q$ is given.

The most important example that we have in mind is the geometric quantization
of the cotangent bundle $T^*Q$ of an orientable Riemannian manifold
$(Q,g)$. This will be described in next section. However, in principle other
examples are possible and thus it is worth to describe the modified
quantization procedure in a more general situation.

\emph{The main idea is that the action of Hamiltonian vector fields of
  quantizable observables on half-forms must be defined through
  $\nabla^Q$-preserving flows.}

\subsubsection{Reducible quantum structures and Liouville manifolds}

Here we will show that there is a close relationship between reducible quantum
structures and Liouville manifolds. We start by giving the necessary
definitions.

\begin{Definition} Let $(M,\omega)$ be a symplectic manifold.  We say that a
  quantum structure $\mathcal Q=(\mathcal L,\langle\, , \rangle^{\mathcal
    L},\nabla^{\mathcal L},P)$ is reducible over a submersion $\pi\colon M\to
  Q$ if both the prequantum structure $(\mathcal L,\langle\, ,
  \rangle^{\mathcal L},\nabla^{\mathcal L})$ and the polarization $P$ are
  reducible over $\pi$ and the prequantum structure has a reduction
  $({L},\langle\, , \rangle,\nabla)$ such that $\nabla^{\mathcal L}$ and
  $\pi^*\nabla$ have isomorphic $P$-polarized sections. In this case we say
  that $({L},\langle\, , \rangle,\nabla)$ is a reduction of the quantum
  structure $\mathcal Q$.
\end{Definition}

\begin{Proposition}\label{Prop:red-quant-pol-compat}
  If a quantization structure $(\mathcal L, {\langle\, , \rangle}^{\mathcal L},
  \nabla^{\mathcal L},P)$ is reducible over a submersion $\pi\colon M\to Q$ and $({L},\langle\, ,
  \rangle,\nabla)$ is a reduction, then
  there exists a $1$-form $\alpha\in\Omega^1(M)$ vanishing on $V(M/Q)$ and a
  closed $2$-form $\bar\omega\in\Omega^2(Q)$ such that
  \begin{equation}
    \omega=\pi^*\bar\omega+d\alpha.\label{eq:32}
  \end{equation}
  Moreover $\bar\omega=i\hbar R[\nabla]$.
\end{Proposition}
\begin{proof} As in the proof of Proposition \ref{Prop:red-quant-struct} one
  has $\nabla^{\mathcal L} -\pi^*\nabla=-\frac{i}{\hbar}\alpha$ for a certain
  $1$-form $\alpha\in\Omega(M)$. Given a vertical vector field
  $V\in\Gamma(M,V(M/Q))=\Gamma(M,P)$ one has
  \begin{displaymath}
    \nabla^{\mathcal L} _V s-\pi^*\nabla_V
    s=-\frac{i}{\hbar}\alpha(V)s
  \end{displaymath}
  and therefore, $\nabla^{\mathcal L} $ and $\pi^*\nabla$ have the same
  $P$-polarized sections if and only if $\alpha(V)=0$. The last claim follows
  from Proposition \ref{Prop:red-quant-struct}.
\end{proof}

According to (see \cite[pag. 234-235]{Va}), a \emph{reducible Liouville
  manifold} is a pair $((M,\omega), P)$ for\-med by a symplectic manifold
$(M,\omega)$ and a real polarization $P$ which is reducible over a submersion
$\pi\colon M\to Q$ such that there exists a $1$-form $\alpha\in \Omega^1(M)$
that vanishes on $P$ and $\omega-d\alpha=\pi^*\bar\omega$ for some
$\omega\in\Omega^2(Q)$. Any such $\alpha\in \Omega^1(Q)$ is called a
(generalized) \emph{Liouville $1$-form} for the Liouville manifold.

We recall that a polarization on a symplectic manifold yields a Liouville
manifold if and only if a certain characteristic class $\mathcal E(P)\in
H^1(M,\mathcal V_P)$, the Euler obstruction of $P$, vanishes. The class is
defined in terms of the $\pi$-relative local system $\mathcal
V_P:=V_{\widehat\nabla^B}(M/Q)$ defined by the parallel sections of $V(M/Q)$
with respect to the Bott connection $\widehat\nabla^B$. In any reducible
Liouville manifold $((M,\omega),P)$ with Liouville form $\alpha$ it holds that
$(M,d\alpha)$ is a symplectic manifold. Therefore, any Liouville manifold is
not compact.

After the above considerations, Proposition \ref{Prop:red-quant-pol-compat}
can be reformulated as follows.
\begin{Proposition} If a quantization $\mathcal Q=(\mathcal L, {\langle\, ,
    \rangle}^{\mathcal L}, \nabla^{\mathcal L},P)$ is reducible over a
  submersion $\pi\colon M\to Q$ and $({L},\langle\, , \rangle,\nabla)$ is a
  reduction, then $((M,\omega),\alpha, P)$ is a Liouville manifold, where
  $\alpha$ is determined by $\nabla^{\mathcal
    L}-\pi^*\nabla=-\frac{i}{\hbar}\alpha$.
\end{Proposition}

We also have the following result.

\begin{Theorem}\label{Theo:red-quantum-structure-pol}
  A quantization $(\mathcal L, {\langle\, , \rangle}^{\mathcal L},
  \nabla^{\mathcal L},P)$ is reducible over a submersion $\pi\colon M\to Q$ if
  and only if the $\pi$-relative connection $\pi_{M/Q}(\nabla^{\mathcal L})$ is
  an absolute parallelism relative to $\pi$; that is, if and only if the
  holonomy groups of $\pi_{M/Q}(\nabla^{\mathcal L})$ along the fibers of $\pi$
  are trivial.
\end{Theorem}
\begin{proof} If we have a reduction $({ L}, \langle\, , \rangle, \nabla)$,
  then $\nabla^{\mathcal L}-\pi^*\nabla=-\frac{i}{\hbar}\alpha$ and therefore
  by Proposition \ref{Prop:red-quant-pol-compat} one has
  \begin{displaymath}
    \pi_{M/Q}(\nabla^{\mathcal L})-\pi_{M/Q}(\pi^*\nabla)=
      - \frac{i}{\hbar}\pi_{M/Q}(\alpha)=0.
\end{displaymath}
Now taking into account Lemma \ref{Lem:canonical-relative-connection-on
  pullbacks} we get $\pi_{M/Q}(\nabla^{\mathcal L})=\widehat\nabla^\pi$ and
therefore $\pi_{M/Q}(\nabla^{\mathcal L})$ is an absolute parallelism relative
to $\pi$, see \cite{tv16}. For the other implication suppose that
$\pi_{M/Q}(\nabla^{\mathcal L})$ is an absolute parallelism on $\mathcal L$,
then it follows from Theorem \ref{Theo:descent-absolute-parallelism} that there
exist a line bundle ${L}\to Q$ such that $\mathcal{L}\simeq \pi^*{L}$ and
$\pi_{M/Q}(\nabla^{\mathcal L})=\widehat\nabla^\pi$ and again by Lemma
\ref{Lem:canonical-relative-connection-on pullbacks} we have
$\pi_{M/Q}(\nabla^{\mathcal L})=\pi_{M/Q}(\pi^*\nabla)$ which shows that the
$1$-form $\nabla^{\mathcal L} - \pi^*\nabla$ vanishes on $V(M/Q)$. Taking into
account Proposition \ref{Prop:red-quant-pol-compat} this finishes the proof.
\end{proof}

\emph{From now on we assume that $(\mathcal L,\langle\, , \rangle^{\mathcal
    L},\nabla^{\mathcal L},P)$ is a quantum structure reducible over a
  submersion $\pi\colon M\to Q$}, with a reduction $({L},\langle\, ,
\rangle,\nabla)$ and $\alpha$ is a Liouville $1$-form.

\subsubsection{Quantizable observables}

As we have seen in Corollary \ref{sec:geom-polar-1}, in the reducible case the
space of polarized sections is defined in a unique way. This allows us to
define the Hilbert space of the quantum theory.

From now on \emph{we assume that $P$ admits a metalinear structure with a line
  bundle of half-forms $K_Q^{1/2}\to Q$}. The Hilbert space of the quantum
theory is, as in standard GQ, built out of the space of polarized
sections. According to Corollary \ref{sec:geom-polar-1}, \emph{the Hilbert
  space of quantum states $\mathcal{H}_P$ is the $L^2$-completion of compactly
  supported sections in $\Gamma(Q,{L}\otimes_Q K^{1/2}_Q)$.}

We recall that in standard GQ, quantizable observables are determined by
requiring that they should preserve the space of polarized sections. Since in
our scheme this space is the same as in GQ, modified quantizable observables
must satisfy the same condition as in GQ: they are functions
$f'\in\mathcal{O}_P\subset\mathcal{C}^{\infty}(M)$ such that $[X_{f'},P]\subset
P$. We recall that $\mathcal{O}_P$ is a Lie subalgebra of the Poisson algebra
$(\cin(M),\{\ ,\ \})$.

Let $((M,\omega),P)$ be a reducible Liouville manifold with associated
submersion $\pi\colon M\to Q$. We observe that for any Liouville $1$-form
$\alpha$ there is a natural lift mapping $\ell_\alpha:\mathfrak{X}(Q)\to
\mathfrak{X}(M)$ such that given $X\in \mathfrak{X}(Q)$, its lift is the unique
vector field $X^{\ell_\alpha}\in \mathfrak{X}(M)$ that projects to $X$ and
satisfies $i_{X^{\ell_\alpha}}d\alpha=-d\alpha(X^{\ell_\alpha})$. In the case
of a cotangent bundle $(M=T^*Q,\omega_0)$, this gives the well known cotangent
lift.

Taking into account that any point of $M$ has a coordinate neighbourhood with
fibered coordinates $\{q_i,p_i\}_{i=1}^n$ for the submersion $\pi\colon M\to Q$
such that $\alpha=p_i dq_i$, see \cite[Proposition 3.5]{Va}, it is easy to
prove the following:
\begin{Theorem}\label{Theo:charact-quantizable-observ} Let $((M,\omega),P)$ be
  a reducible Liouville manifold with associated submersion $\pi\colon M\to
  Q$. If $\alpha$ is a Liouville $1$-form $\alpha$ with associated lift
  $\ell_\alpha:\mathfrak{X}(Q)\to \mathfrak{X}(M)$, then one has $$\mathcal
  O_P=\{f'\in \cin(M)\colon f'=\pi^*{f}+\alpha(X^{\ell_\alpha}),
  {f}\in\cin(Q), X\in\mathfrak{X}(Q)\}.$$ Moreover, for any
  ${f}\in\cin(Q)$ and $X\in\mathfrak{X}(Q)$, the associated
  Hamiltonian vector fields $X_{\pi^*{f}}$,
  $X_{\alpha(X^{\ell_\alpha})}$ on $(M,\omega)$ verify that
  \begin{displaymath}
  X_{\pi^*{f}},\ \text{and}\
  X_{\alpha(X^{\ell_\alpha})}-X^{\ell_\alpha}\ \text{are}\
  \pi\text{-vertical}.
\end{displaymath}
Therefore, for any $f'=\pi^*{f}+\alpha(X^{\ell_\alpha})\in\mathcal O_P$
the Hamiltonian vector field $X_{f}\in\mathfrak{X}(M)$ is
projectable and its projection is ${X}\in\mathfrak{X}(Q)$.
\end{Theorem}

In the reducible case the space of quantizable observables $\mathcal O_P$ can
be identified with $\cin(Q)\times\mathfrak{X}(Q)$ and thus we can transfer to
it the Lie algebra structure of $\mathcal O_P$. An easy computation proves the
following result.

\begin{Proposition} Let $((M,\omega),P)$ be a reducible Liouville manifold with
  associated submersion $\pi\colon M\to Q$ and Liouville form $\alpha$. Let
  $\bar\omega\in\Omega^2(Q)$ is a closed form such that
  $\omega=d\alpha+\pi^*\bar\omega$, then the Lie algebra structure of $\mathcal
  O_P$ is given by
  \begin{displaymath}
    \{f_1',f_2'\}=
    \{({f}_1,{X}_1),({f}_2,{X}_2)\}=
    ({X}_1{f}_2-{X}_2{f}_1
    -\bar\omega({X}_1,{X}_2),[{X}_1,{X}_2]),
  \end{displaymath}
  for any
  $f_1'=({f}_1,{X}_1),f_2'=({f}_2,
  {X}_2)\in\cin(Q)\times\mathfrak{X}(Q)$.
\end{Proposition}



\subsubsection{Quantization of classical observables}

Here, quantum operators are modified according to our guiding idea. We denote
by $\tilde{X}^{\nabla^Q}_{f} $ the vector field on the tensor product bundle
$\mathcal L\otimes_M N^{1/2}_P$ naturally induced by the vector fields
$X^{\ell'}_{f'} := h^{\nabla^{\mathcal L}}(X_{f}) + i\frac{f}{\hbar}E^{\mathcal
  L}$, where $E^{\mathcal L}$ is the Euler vector field of the Line bundle
$\mathcal L$, and by $X^{Q}_{f'} := h^{\pi^*(\nabla^Q)}(X_{f'})$.

\begin{Definition}[Modified geometric quantization]
  \label{sec:geom-polar-modif-1}
  A quantizable observable $'\in\mathcal{O}_P$ yields a quantum operator
  \begin{equation}
    \widehat{f'}\colon \Gamma(M,\mathcal L\otimes N_P^{1/2})\to
    \Gamma(M,\mathcal L\otimes N_P^{1/2})\label{eq:33}
  \end{equation}
  that on a decomposable half-form ${\tilde\psi}=s'\otimes\sqrt{\nu}\in
  \Gamma(M,\mathcal L\otimes N_P^{1/2})$ acts by
  \begin{equation}\label{eq:29}
    \widehat{f'}(\tilde{\psi}) = i\hbar L_{\tilde{X}^Q_{f'}}{\tilde\psi} =
    i\hbar(L_{X^{\ell'}_{f'}}s'\otimes\sqrt{\nu} +
    s'\otimes L_{X^Q_{f'}}\pi^*\sqrt{\nu}).
  \end{equation}
\end{Definition}

In the above definition we operate by projectable vector fields. Indeed, given
$f'\in\mathcal O_P$ such that $f'=\pi^*{f}+\alpha(X^{\ell_\alpha})$ with
${f}\in\cin(Q)$ and $X\in\mathfrak{X}(Q)$, one
has
\begin{align*}
  X_{f'}^{\ell'}&=h^{\nabla^{\mathcal L}}(X_{f'})+i\frac{f'}{\hbar}E^{\mathcal
    L}
  =h^{\pi^*\nabla-\frac{i}{\hbar}\alpha}(X_{f})+i\frac{f}{\hbar}E^{\mathcal L}=
  \\
  &=h^{\pi^*\nabla}(X_{f'})-\frac{i}{\hbar}\alpha(X_{f'})E^{\mathcal L} +
  i\frac{f'}{\hbar}E^{\mathcal L}.
\end{align*}
Since $\alpha$ vanishes on vertical vector fields, by Theorem
\ref{Theo:charact-quantizable-observ} we have
$\alpha(X_{f'})=\alpha(X^{\ell_\alpha})$ and therefore
\begin{equation}
X_{f'}^{\ell'}=
h^{\pi^*\nabla}(X_{f'})+i\frac{\pi^*{f}}{\hbar}E^{\mathcal L}.\label{eq:34}
\end{equation}
Moreover, since $X_{f'}$ projects to $X$, and the Euler vector field
$E^{\mathcal L}$ of $\mathcal L$ projects to the Euler vector field ${E}$ of
${L}$, it follows that $h^{\pi^*\nabla^Q}(X_{f})$ projects to $h^{\nabla^Q}(X)$
and $i\frac{\pi^*{f}}{\hbar}E^{\mathcal L}$ projects to
$i\frac{{f}}{\hbar}{E}$. Hence we have proved that
$X_{f'}^{\ell'}\in\mathfrak{X}(\mathcal L)$ projects to
$X_{f'}^\ell:=h^{\nabla}(X)+i\frac{f}{\hbar}E\in\mathfrak{X}(L)$. In a similar
way ${X}^{Q}_{f'} = h^{\pi^*(\nabla^Q)}(X_{f'})\in\mathfrak{X}(N_P^{1/2})$
projects to ${X}^{Q} := h^{\nabla^Q}({X})\in\mathfrak{X}(K_Q^{1/2})$. We denote
by ${X}^{Q}_{f'} $ the vector field on the tensor product bundle $L\otimes_M
K_Q^{1/2}$ naturally induced by $X^\ell_{f'}$ and ${X}^Q$. The above discussion
proves the following:
\begin{Proposition}
  \label{pro:quantum-operator-action-on-polarized-wave-functions}
  For any $f'=\pi^*{f}+\alpha(X^{\ell_\alpha})\in\mathcal O_P$ with
  ${f}\in\cin(Q)$ and ${X}\in\mathfrak{X}(Q)$ the action of
  its associated operator $\widehat{f'}$ on a decomposable half-form
  ${\tilde\psi}=\pi^*\psi$ with $\psi=s\otimes\sqrt{\nu}\in\cH_P$ is given by
  \begin{align}
    \label{eq:30}\widehat{f'}(\tilde{\psi}) = i\hbar
    \pi^*(L_{{X}^{Q}_{f'}}\psi) = & i\hbar
    \pi^*(L_{{X}_{f'}^\ell}s\otimes\sqrt{\nu} + s\otimes
    L_{{X}^{Q}}\sqrt{\nu})
    \\ = & i\hbar \pi^*((\nabla_{{X}}s+i\frac{f}{\hbar}s)\otimes\sqrt{\nu} +
    s\otimes \nabla^Q_{{X}}\sqrt{\nu}).
  \end{align}
  Therefore the modified quantum operator $\widehat f'$ defined in
  \eqref{eq:29} preserves $\mathcal H_P$ (i.e
  $\widehat{f'}(\mathcal{H}_P)\subset \mathcal{H}_P$), and its action on wave
  functions $\psi\in \mathcal H_P$ is given by $\widehat f'(\psi)=i\hbar
  L_{{X}^{Q}_{f'}}\psi$.
\end{Proposition}

Taking into account the identification $\mathcal
O_P=\cin(Q)\times\mathfrak{X}(X)$ proved in Theorem
\ref{Theo:charact-quantizable-observ}, for every
$f'=({f},{X})\in\mathcal O_P$ we define an
operator
\begin{equation}
  \rho^L_{f'}:=L_{X_{f'}^\ell}\colon \Gamma(Q,L)\to \Gamma(Q,L),
  \quad
  \rho^L_{{f'}}(s)=\nabla_X s+i\frac{{f}}{\hbar}s.
  \label{eq:35}
\end{equation}
Analogously we define
\begin{equation}
\rho_{f'} \colon
\Gamma(Q,L\otimes K_Q^{1/2})\to \Gamma(Q,L\otimes K_Q^{1/2}),
\quad
\rho_{f'}(\psi)=\rho_{f'}^L(s)\otimes\sqrt{\nu}+s\otimes
\nabla^Q_{X}\sqrt{\nu}.
\label{eq:36}
\end{equation}
The action of the quantum operator associated to $f'=(f,X)\in\mathcal O_P$
simply reads $\widehat f'(\psi)=i\hbar \rho_{f'}(\psi).$ In order for the
modified geometric quantization to be physically acceptable we have
to check that the fundamental commutation identity~\eqref{eq:27} still
holds. As we will shortly see this is true if we require the flatness of
$\nabla^Q$, regarded as a connection on $K_Q^{1/2}$. For the same reason we
also need to check that the quantum operators are symmetric. In order to do
this, since they are differential operators and the space of compactly
supported $\cin$ wave functions $\Gamma_0(Q,L\otimes K_Q^{1/2})$ is dense in
the Hilbert space $\mathcal H_P$, it is enough to check that they are formally
self-adjoint.
\begin{Theorem}\label{sec:quant-class-observ}
  Let $f_1'=({f}_1,{X}_1)$,
  $f_2'=({f}_2,{X}_2)\in\mathcal{O}_P$, and suppose that the
  connection $\nabla^Q$ on $K^{1/2}_Q\to Q$ is flat. Then
  \begin{equation}
    \widehat{\{f_1',f_2'\}}=i\hbar[\widehat f_1',\widehat f_2']. \label{eq:28} 
  \end{equation}
  For every $f'=({f},{X})\in\mathcal O_P$ the associated
  quantum operator $\widehat f'\colon \mathcal H_P\to\mathcal H_P$ is symmetric
  if the linear operator $A_X=L_X-\nabla_{{X}}^Q$ vanishes on the line bundle
  $|K_Q|$ of densities. If in addition $X$ is complete then $\widehat f'$ is
  self-adjoint.
\end{Theorem}
\begin{proof}
  One easily checks that $[X^\ell_{f_1'},X^\ell_{f_2'}] =
  X^\ell_{\{f_1',f_2'\}}$. Then it is straightforward to prove that for a
  decomposable wave function $\psi=s\otimes\sqrt{\nu}$ one has
  \begin{align*}
    [\widehat f_1',\widehat f_2']({\psi}) = & (i\hbar)^2
    (L_{X^\ell_{\{f_1',f_2'\}}}s\otimes\sqrt{\nu} + s\otimes
    [\nabla^Q_{{X}_1},\nabla^Q_{\bar{X}_2}]\sqrt{\nu})
    \\
    = & (i\hbar)^2 (\rho^L_{\{f_1',f_2'\}}s\otimes\sqrt{\nu} + s\otimes
    \nabla^Q_{[X_1,X_2]}\sqrt{\nu})
    \\
    = & (i\hbar)^2\rho_{\{f_1',f_2'\}}(\psi)= i\hbar
    \widehat{\{f_1',f_2'\}}(\psi).
  \end{align*}

  If $f'=({f},{X})\in\mathcal O_P$ and
  $\psi_1=s_1\otimes\sqrt{\nu_1}$, $\psi_2=s_2\otimes\sqrt{\nu_2}\in
  \Gamma(Q,L\otimes K_Q^{1/2})$, then a straightforward computation
  shows
  \begin{equation}
    \langle\widehat f'(\psi_1),\psi_2\rangle\bar{}=\int_Q\langle\widehat
    f'(\psi_1),\psi_2\rangle=\langle\psi_1,\widehat f'(\psi_2)\rangle_Q+\int_Q
    \nabla^Q_{{X}}(\langle s_1,s_2\rangle\langle
    \sqrt{\nu_1},\sqrt{\nu_2}\rangle_P),\label{eq:37}
  \end{equation}
  where $\nabla^Q_{{X}}(\langle s_1,s_2\rangle\langle
  \sqrt{\nu_1},\sqrt{\nu_2}\rangle_P)$ denotes the covariant derivative of the
  density $\eta=\langle s_1,s_2\rangle\langle
  \sqrt{\nu_1},\sqrt{\nu_2}\rangle_P$. If $A_X=0$ then we
  have $\int_Q\nabla_X^Q\eta=\int_Q L_X\eta$. Now if $\psi_1$ is compactly
  supported then the same is true for $\eta$ and it is well known that in this
  case $\int_Q L_X\eta=0$, and so the proof is finished.
\end{proof}

In the first part of the previous proof the role played by the flatness of
$\nabla^Q$ is clearly fundamental. Notice however that, for several good
reasons, the flatness condition on $\nabla^Q$ should not be regarded as a
strong requirement:
\begin{itemize}
\item in GQ this prescription is fulfilled as we act on half-forms by a Lie
  derivative that behaves like the covariant derivative of a flat connection
  \cite[p.\ 185]{Woo92}.
\item in the main example of quantization of the phase space of a particle
  moving on a Riemannian manifold the flatness of $\nabla^Q$ holds true even if
  we use the connection induced by the Levi-Civita connection (see
  Remark~\ref{sec:conn-half-forms}).
\end{itemize}

The most delicate condition is the divergence-free requirement on the vector
field $X$.  Here we make the following important remarks:
\begin{itemize}
\item The most important class of examples for applications in Physics is the
  quantization of natural mechanical systems. Here, as we will see in the next
  section, the basic examples of physical observables satisfy the
  divergence-free property. Moreover, it is well-known that the space of
  diver\-gence-free vector fields is infinite-dimensional, and therefore this
  means that the space of observables that can be straightforwardly quantized
  is large.
\item In natural mechanical systems the Riemannian metric is supposed to play
  an important role and as a matter of fact it enters in the Schr\"{o}dinger
  operator. This means that extending the role of the metric should be quite a
  natural assumption. Just as an example, every Killing (\emph{i.e.} metric
  preserving) vector field $X$ is divergence-free, and in view of the above
  remark this is a physically meaningful class of vector fields.
\end{itemize}

The quantization of non straightforwardly quantizable observables in our
modified quantization method is performed by changing the BKS method
according to our guiding principles.

\begin{Definition}[Modified BKS method]
  \label{sec:geom-polar-modif-2}
  A non straightforwardly quantizable observable $H\in\cC^\infty(M)$ yields a
  quantum operator $\widehat{H}$ which is defined on decomposable wave
  functions ${\psi}=s\otimes\eta\in\cH$ as
  \begin{displaymath}
    \widehat{H}(\tilde{\psi}) =
    i\hbar \frac{d}{dt}\cH_{0t}\circ
    \tilde{\phi}^{\nabla^Q,H}_{t,*}(\tilde{\psi})\big|_{t=0},
  \end{displaymath}
  where
  \begin{displaymath}
    \tilde{\phi}^{\nabla^Q,H}_{t,*}(\tilde{\psi}) =
    \phi^{H,L}_{t,*}(\psi)\otimes\phi^{\nabla^Q,H}_{t,*}(\sqrt{\nu}),
  \end{displaymath}
  $\tilde{\phi}^{\nabla^Q,H}_t$ is the flow of $\tilde{X}^{\nabla^Q}_H$ and
  $\phi^{\nabla^Q,H}_t$ is the flow of $X^{\nabla^Q}_H$.
\end{Definition}

\section{Modified geometric quantization on Riemannian manifolds}
\label{sec:concr-exampl-quant}

In this section we specialize our constructions to the case of a classical
mechanical system for a particle, modelled by a Riemannian manifold $(Q,g)$.
Its phase space is the symplectic manifold $(T^*Q,\omega_0)$, where
$\omega_0=d\theta_0$ is the canonical symplectic form, the differential of the
canonical Liouville form $\theta_0$.

We derive the quantum operators for the most common observables with our
modified quantization method, showing that it leads to interesting consequences
that will be analyzed in the conclusive section.

We assume that $Q$ is orientable, so that $\wedge^n T^*Q$ is a trivial real
line bundle.  We will denote by $(q^i)$ local coordinates on $Q$; consequently
$(q^i,p_i)$ will denote coordinates on $T^*Q$.

We refer to \cite{Sni80} for a detailed derivation of the energy operator in
this situation within the framework of the standard GQ (see p.\ 120, Section
7.2, or p.\ 180, Section 10.1 for the case with a nonzero electromagnetic
field).

We assume that the topological conditions for the existence of a prequantum
structure on $(T^*Q,\omega_0)$ are fulfilled (see
Section~\ref{sec:prel-geom-quant}). It can be proved that the topological
triviality of the typical fiber of $T^*Q\to Q$ implies that \emph{all
  prequantum structures are reducible}. So, we may assume without loss of
generality a quantum bundle $\mathcal L\to T^*Q$ endowed with a connection
which are the pull-back, respectively, of a Hermitian complex line bundle
${L}\to Q$ and a connection $\nabla$ on this bundle.

The polarization $P$ that we choose is the vertical one, \emph{i.e.}
$P=VT^*Q$ which is locally spanned by the vector
fields $\partial/\partial p_i$.  Since $Q$ is orientable, its
determinant bundle admits a square root $K_Q^{1/2}=\sqrt{\wedge^n_\C T^*Q}$, and
there exists a metaplectic structure: $N_P^{1/2}=\tau^\vee_Q(K_Q^{1/2})$.

The canonical bundle $N_P^{1/2}$ is endowed with the pull-back
$\tau^\vee_Q(\nabla^g)$ of the canonical connection $\nabla^Q=\nabla^g$ on
$K_Q^{1/2}$ defined by the Levi-Civita connection on $\tau^\vee_Q\colon
T^*Q\to Q$ using the natural procedures described in
Subsection~\ref{sec:math-prel} (see~\eqref{eq:13}).

The Hilbert space $\mathcal{H}_P$ of quantum states is given in both GQ and the
modified GQ by the $L^2$-completion of compactly supported polarized sections
in $\Gamma_P(T^*Q,\mathcal L)=\Gamma(Q,{L}\otimes_Q K_Q^{1/2})$ (see
Corollary~\ref{sec:geom-polar-1}).

\begin{Proposition} If ${\psi}=s\otimes\sqrt{\nu_g}$ is a wave function,
  $(q",p_i)$ are canonical coordinates and $X=f^i\frac{\partial}{\partial
    q^i}$ is a divergence-free vector field, then the modified geometric
  quantization procedure yields
  \begin{gather}
    \label{eq:17b}
    \widehat{q^i}(\psi) = i\hbar q^is\otimes\sqrt{\nu_g}, \\ \label{eq:11}
    \widehat{f^ip_i}(\psi) = - i\hbar
    f^i\pd{}{p_i}\psi^0b_0\otimes\sqrt{\nu_g},
  \end{gather}
\end{Proposition}
\begin{proof}
  According to Definition~\ref{sec:geom-polar-modif-1} and Proposition \ref{pro:quantum-operator-action-on-polarized-wave-functions} the quantum operator
  $\widehat f'$ acts on a wave function ${\psi}=s\otimes\sqrt{\nu_g}$ as
  \begin{equation}\label{eq:18}
    \hat{f'}({\psi}) = i\hbar(L_{X^\ell_{f'}}s\otimes\sqrt{\nu_g} +
    s\otimes L_{{X}^g_f}\sqrt{\nu_g}),
  \end{equation}
  where ${X}^g_f$ here stands for
  ${X}^{Q}_f$. Using~\eqref{eq:10}, we have
  \begin{displaymath}
    L_{{X}^g_f}\sqrt{\nu_g} = \nabla_{X}\sqrt{\nu_g} = 0.
  \end{displaymath}
  So only the first summand in~\eqref{eq:18} contributes to the quantum
  operator and it is easy to prove the result.
\end{proof}
Of course, the vanishing of the divergence term is due to the fact that we have
\emph{chosen} a divergence-free vector field $X$.

\begin{Remark}
  The natural observables like positions $q^i$ and momenta $p_i$ fulfill the
  divergence-free condition. We would like to stress that also the angular
  momentum fulfills the same conditions, and hence is also straightforwardly
  quantizable in the modified quantization (see, \emph{e.g.},
  \cite{ModTejVit08b}).
\end{Remark}

The most interesting result concerns the quantization of the Hamiltonian $H$.
\begin{Theorem}
  The modified BKS method yields the quantum energy operator
  \begin{equation}
    \label{eq:19}
    \widehat{H} =  \left(- \frac{\hbar^2}{2}\Delta(\psi) + V\psi\right)
    \otimes \sqrt{\nu_g}.
  \end{equation}
\end{Theorem}
\begin{proof}
  We shall evaluate the pairing~\eqref{eq:15}. Following \cite[p.\ 202]{Woo92},
  we have
  \begin{equation}
    \label{eq:22}
    (\tilde{\psi}_1,\cH_{0t}(\tilde{\psi}_2))=
    \int_{T^*Q}\psi_1\overline{\psi_2}e^{-itH/\hbar}
    \sqrt{(\nu_g,\phi^{g,H}_{t,*}\nu_g)}
  \end{equation}
  where $\phi^{g,H}_t$ denotes the flow of $X^g_H$. Since in the modified BKS
  procedure we replaced the flow of $X_H$ by the flow of $X^g_H$ we have that
  \begin{equation}
    \label{eq:23}
    \sqrt{(\nu_g,\phi^{g,H}_{t,*}\nu_g)} = \nu_g
  \end{equation}
  as the parallel transport preserves the metric volume.

  The left-hand side of~\eqref{eq:23} was the source of the scalar curvature in
  the quantum energy operator obtained through the BKS procedure. This means
  that in the modified BKS procedure the scalar curvature term is not present
  as the time derivative of $\nu_g$ vanishes.
\end{proof}

\begin{Corollary}
  Let $\psi = s\otimes\sqrt{\nu_g}\in\cH_P$. Then
  \begin{displaymath}
    \Delta(s)\otimes\sqrt{\nu_g} = \tilde{\Delta}({\psi}),
  \end{displaymath}
  where $\Delta$ and $\tilde{\Delta}$ are the Bochner Laplacians of the line
  bundles ${L}\to Q$ and ${L}\otimes_Q K^{1/2}_Q\to Q$, respectively.
\end{Corollary}
\begin{proof}
  The proof is easily achieved by recalling that the Bochner Laplacian is just
  the double covariant derivative followed by a metric contraction. Since the
  covariant derivative in the tensor product bundle ${L}\otimes_Q
  K^{1/2}_Q\to Q$ annihilates the factor $\sqrt{\nu_g}$, the result follows.
\end{proof}

\begin{Remark}\label{sec:modif-geom-quant}
  In a Riemannian manifold $(Q,g)$ the scalar curvature $r_g$ measures the
  volume distorsion (to the second order) of balls of small radius in $Q$
  compared to Euclidean balls of the same radius \cite[p.\ 168, Theorem
  3.98]{GHL}. However if we move the volume along the flow of the parallel
  transport, there is no distorsion.
\end{Remark}

\section{Discussion}

In the above sections we have demonstrated that the wide set of choices that
build up one representation in GQ can be further enlarged if we allow for the
possibility to choose a connection on the polarization bundle such that the
induced connection on the half-form bundle is flat. The additional restriction
that we have obtained for the coefficients of momenta in quantizable
observables, the divergence free condition, does not seem to be too strong
since the basic observables of position, momenta and angular momenta fulfill
the condition.  However, there are many other theories that yield a quantum
energy operator with scalar curvature as a multiplicative operator (see the
Introduction for details).  In particular we may mention Feynman path
integral~\cite{dew,Che,BAL73} or Weyl quantization~\cite{LQ92,Lan98}.

It is not possible to prove in one single paper that a modification of Feynman
path integral or Weyl quantization might yield a Schr\"odinger operator with no
scalar curvature term. However, after the above results, it seems to us that
there exists a possibility that the Feynman or Weyl quantization procedures
might be modified along the above ideas to yield $k=0$, or even any other
nonzero value of $k$ like in \cite{JanMod02}, as was already remarked in
\cite{dew}.

The technical remark~\ref{sec:modif-geom-quant} can help understanding how
to carry on a quantization procedure in a way that would preserve the metric
volume.

We deliberately ignored the important topic of complex polarizations. The
possibility that a polarization has a natural geometric structure inherited
from a metric or by a natural connection in some related space is concrete, and
it would be at least interesting to explore it by analogy with our previous
guidelines in order to see if any interesting consequences would appear.

\paragraph{Acknowledgements.} We would like to thank M. Modugno, D. Canarutto,
J. Jany\v ska and J. Kijowski for stimulating discussions.

We acknowledge the support of Ministerio de Econom\'\i a y Competitividad,
Gobierno de Espa{\~n}a, Research Grant MTM2013-45935-P, Dipartimento di
Matematica e Fisica `E. De Giorgi' dell'Universit\`a del Salento, Gruppo
Nazionale di Fisica Matematica dell'Istituto Nazionale di Alta Matematica
\url{http://www.altamatematica.it} and INFN -- IS CSN4 \emph{Mathematical
  Methods of Nonlinear Physics}.


\end{document}